\documentclass[12pt, a4paper]{article}
\usepackage[utf8]{inputenc}
\usepackage{dcolumn,lscape}
\usepackage{amsmath,longtable,multicol,dcolumn,tabularx,graphicx,amssymb}
\usepackage{exscale,amsthm, multirow, rotating}
\usepackage{natbib}
\usepackage{bbm}
\usepackage{afterpage, placeins}
\usepackage{tikz,dsfont,float}
\usepackage[caption = false]{subfig}
\newcommand{\xvec}{\boldsymbol}
\newcommand{\xmat}{\mathbf}

\newtheorem{proposition}{Proposition}

\usetikzlibrary{calc}
\usetikzlibrary{positioning}
\begin{document}

\def\spacingset#1{\renewcommand{\baselinestretch}%
{#1}\small\normalsize} \spacingset{1}


\title{\bf Estimation of Spatial and Temporal Autoregressive Effects using LASSO -- An Example of Hourly Particulate Matter Concentrations}
  \author{Elkanah Nyabuto\footnote{email: e.nyabuto.1@research.gla.ac.uk}\\
    \small{University of Glasgow, UK}\\
    Philipp Otto\footnote{email: philipp.otto@glasgow.ac.uk}\\
    \small{University of Glasgow, UK}\\
    Yarema Okhrin\footnote{email: yarema.okhrin@uni-a.de}\\
    \small{University of Augsburg, Germany}}
    \date{}
  \maketitle

\begin{abstract}
\noindent We present an estimation procedure of spatial and temporal effects in spatiotemporal autoregressive panel data models using the Least Absolute Shrinkage and Selection Operator, LASSO \citep{tibshirani1996regression}. We assume that the spatiotemporal panel is drawn from a univariate random process and that the data follows a spatiotemporal autoregressive process which includes a regressive term with space-/time-varying exogenous regressor, a temporal autoregressive term and a spatial autoregressive term with an unknown weights matrix. The aim is to estimate this weight matrix alongside other parameters using a constraint penalised maximum likelihood estimator. Monte Carlo simulations showed a good performance with the accuracy increasing with an increasing number of time points. The use of the LASSO technique also consistently distinguishes between meaningful relationships (non-zeros) from those that are not (existing zeros) in both the spatial weights and other parameters. This regularised estimation procedure is applied to hourly particulate matter concentrations (PM$_{10}$) in the Bavaria region, Germany for the years 2005 to 2020. Results show some stations with a high spatial dependency, resulting in a greater influence of PM concentrations in neighbouring monitoring stations. The LASSO technique proved to produce a sparse weights matrix by shrinking some weights to zero, hence improving the interpretability of the PM concentration dependencies across measurement stations in Bavaria. 
\end{abstract}

\noindent%
\vfill

\spacingset{1.45} 

\section{Introduction}
Spatial and temporal autoregressive models have become increasingly popular for analysing spatiotemporal data in various fields in recent years. The analysis of such data requires accounting for spatial and temporal dependencies which are modelled by spatiotemporal autoregressive models \citep{cressie2015statistics}. A fundamental element of these models is the incorporation of a spatial weights matrix, denoted as $\xmat{W}$, that characterises relationships between different spatial units \citep{anselin1988spatial}. The spatial weights matrix is a square matrix of size $N$, with $N$ being the number of spatial units in the model. This matrix defines the strength of connections or influences between neighbouring spatial units. The weight matrix choice depends on the underlying spatial process and/or the nature of interactions between spatial units \citep{lesage2009introduction}. Commonly used weight matrices include those based on contiguity (considering neighbours) or distance decay, weighing influences based on proximity \citep{ward2018spatial}, or other methods tailored to the specific characteristics of the spatiotemporal data under investigation.  However, this matrix is usually unknown for empirical applications \citep{gibbons2012mostly}. The thoughtful selection of this matrix is crucial, as it directly impacts the accuracy of parameter estimates and the model's ability to capture spatial dependencies effectively \citep{ward2018spatial}.

To address this challenge, \cite{bhattacharjee2006estimation} focused on the estimation of the spatial weights matrix within the context of diffusion in housing demand, providing a practical illustration of the implications of spatial dependencies. Further, \cite{getis2004constructing} explored the construction of spatial weights matrices through local statistics, offering alternative approaches for defining spatial relationships. Later, \cite{bhattacharjee2013estimation} introduced a methodology for estimating the spatial weight matrix under structural constraints, providing foundational insights into the relationships between spatial dependencies and structural considerations. \cite{ahrens2015two} also contributed by developing a two-step LASSO estimation method for the spatial weights matrix, enhancing precision through the sparsity-inducing properties of LASSO, a regression penalization technique introduced by \citep{tibshirani1996regression}. \cite{zhang2018spatial} worked into spatial weights matrix selection from a rich set of candidate matrices using a regularised approach and model averaging, enriching the toolkit for model selection and considering spatial dependencies. Moreover, \cite{otto2023estimation} proposed an adaptive LASSO approach for estimating the entire weights matrix without considering further structural constraints. Their approach allowed for an unknown number of structural breaks in the mean. \cite{merk2022estimation} considered a purely spatial setting, for which they also estimate the weights matrix using a LASSO approach and re-sampling sub-blocks of the entire spatial domain. These methods allow for efficient and effective variable selection in high-dimensional settings, making them ideal for estimating the weights matrix and another model parameter in the spatiotemporal autoregressive panel data models case. 

In this paper, we consider a spatiotemporal autoregressive model that includes a regressive term with space-time-varying exogenous regressors, higher-order temporal autoregressive lags, and a spatial autoregressive term with an unknown weights matrix. We propose to estimate $\xmat{W}$, alongside other model parameters using a constraint penalized maximum likelihood estimator with $l_1$ penalties.  The LASSO method under the sparsity assumption, deals with high dimensional settings where the number of regressors is more than the sample size \citep[see][for an overview]{friedman2010note}. The $l_1$-penalization, sets some of the coefficients to exactly zero, thus good for model selection \citep{tibshirani1996regression}. In the context of estimating the weights matrix, the  $l_1$ penalties impose sparsity on the weights matrix, improving prediction accuracy and interpretation \citep{wang2008note}. The LASSO estimator also provides a trade-off between model complexity and prediction accuracy by selecting relevant weights and shrinking the irrelevant ones towards zero hence resulting in a sparse weights matrix \citep{zhang2010nearly}. By estimating all elements of the weights matrix, as well as the autoregressive coefficients and the error variance, our method provides a comprehensive solution for modelling spatiotemporal autoregressive effects while avoiding overfitting and selecting a parsimonious model. This approach provides a powerful tool for estimating spatiotemporal autoregressive effects with high-dimensional covariates and large datasets e.g., for modelling hourly particulate matter concentrations.

The rest of the paper is organised as follows: Section 2 presents the spatiotemporal dynamic panel data model with its properties described. Section 3 discusses the estimation procedures employed in this paper. The constrained LASSO penalisation technique is also discussed, along with the selection consistency of the parameters. Section 4 presents the Monte Carlo simulations, which evaluate the performance of the model under different spatial and temporal resolutions. Section 5 presents the empirical illustration of the model for modelling air quality in Bavaria, Germany, and the discussion of the results thereafter.

\section{Spatiotemporal Dynamic Panel Data Model}
Spatiotemporal autoregressive models are used to capture and understand the relationships and dependencies among observations across both space and time. In this section, we introduce and discuss the spatiotemporal model used in this paper and important properties that will be needed in the coming sections. 

We assume that the spatiotemporal panel is drawn from a univariate random process $\{Y_t(\xvec{s}): t = 1, \ldots, T; \xvec{s} \in D_s\}$ with a finite set $D_s$ of $n$ locations $\xvec{s}_1, \ldots, \xvec{s}_n$. Moreover, let $\xvec{Y}_t$ the stacked vector of all random variables across space at time point $t$, i.e.,  $\xvec{Y}_{t} = (Y_t(\xvec{s}_1), \ldots, Y_t(\xvec{s}_n))'$. We consider data that follows a spatiotemporal dynamic panel data model 
\begin{equation*}
\xvec{Y}_{t} = \xmat{X}_{t} \xvec{\beta} + \sum_{p = 1}^{P} \mathbf{\Phi}_{p} \xvec{Y}_{t-p} + \xmat{W}\xvec{Y}_{t} + \xvec{\varepsilon}_{t}.
\end{equation*}
The model includes a regressive term with space-/time-varying $k-$dimensional exogenous regressor $\xmat{X}_t$ with $\xvec{\beta}=(\beta_1,\beta_2,\ldots,\beta_k)'$ regression coefficients, a temporal autoregressive term of order $P$, and a spatial autoregressive term with an unknown weights matrix $\xmat{W}$. The $\mathbf{\Phi}_{p}=\text{diag}(\phi_{p}(\xvec{s}_1), \ldots, \phi_{p}(\xvec{s}_n))$ are the temporal model coefficients for the previous $p$ realisations of $\xvec{Y}_{t}$ at time $t$ and that $\xvec{Y}_{t-p}$ is the $p^{th}$ previous realisations at time $t$. Here, we initially focus on a diagonal structure of these matrices; hence, spatial interactions are fully modelled by the contemporaneous spatial autoregressive term. In general, spatiotemporal spillover effects can be directly incorporated into the $\mathbf{\Phi}_{p}$ matrices, which we will leave for future research. Note that the model still includes spatiotemporal effects, as we will explain in further detail below. 

In spatial econometrics, $\xmat{W}$ is typically replaced by a product of an unknown coefficient $\rho$ with a prespecified -- assumed to be known -- matrix $\tilde{\xmat{W}}$, e.g., based on neighbourhood relations or inverse distance weights. However, this matrix is usually unknown for empirical applications \citep{gibbons2012mostly}; therefore, we consider the full matrix $\xmat{W}$ to be unknown. The goal is to estimate this weight matrix along with all other parameters. The random errors are assumed to be independent and identically distributed with zero mean and constant variance $\sigma_\varepsilon^2$ (across space and time). In the reduced form, we get that
\begin{equation*}\label{eq:reducedform}
\xvec{Y}_{t} = (\xmat{I} - \xmat{W})^{-1}\left( \xmat{X}_t \xvec{\beta} + \sum_{p = 1}^{P} \mathbf{\Phi}_p \xvec{Y}_{t-p}  + \xvec{\varepsilon}_t \right) \, .
\end{equation*}
Thus, it is important to ensure that $(\xmat{I} - \xmat{W})$ is non-singular. Furthermore, it is important to note that spatiotemporal autoregressive effects are also included in this model because
\begin{eqnarray*}
\xvec{Y}_{t} &  =  & \sum_{p = 1}^{P} (\xmat{I} - \xmat{W})^{-1} \mathbf{\Phi}_p \xvec{Y}_{t-p}  + (\xmat{I} - \xmat{W})^{-1}\left( \xmat{X}_t \xvec{\beta} + \xvec{\varepsilon}_t \right) \, \\
             &  =  & \sum_{p = 1}^{P} \sum_{k = 0}^{\infty} \xmat{W}^{k}\mathbf{\Phi}_p \xvec{Y}_{t-p} + (\xmat{I} - \xmat{W})^{-1}\left( \xmat{X}_t \xvec{\beta} + \xvec{\varepsilon}_t \right) \,
\end{eqnarray*}
with $\xmat{W}^0 = \xmat{I}$. If the series $\xmat{I} + \xmat{W} + \xmat{W}^2 + ...$ converges, $(\xmat{I} - \xmat{W})$ is invertible and the process is stationary. However, to ensure stationarity, restrictions on  $\xmat{\Phi}_p$ matrices are essential by ensuring its diagonal elements are less than $1$. This condition ensures that the influence of past values on current values diminishes over time, preventing explosive behaviour in the model and, thus, allowing for consistent and stable model predictions over time.  

\begin{proposition}\label{prop:stationarity}
Let $\xmat{W}$ be an $n \times n$ square matrix representing spatial relationships among $n$ spatial units, which is uniformly bounded in row and column sums, and let $\{\xmat{\Phi}_p : p = 1, ..., P\}$ be diagonal matrices of temporal coefficients for the same $n$ spatial units and $P$ temporal lags. The process is stationary and stable across time/space, if the inverse $(\xmat{I}-\xmat{W})^{-1}$ exists and the operator norm of the product $\sum_{p = 1}^{P}(\xmat{I}-\xmat{W})^{-1}\xmat{\Phi}_p$ is less than 1.

\end{proposition}
\begin{proof}
Express the series $\sum_{p = 1}^{P}\xmat{\Phi}_p(\xmat{I}-\xmat{W})^{-1}$ as the sum of its terms,
$$\sum_{p = 1}^{P}(\xmat{I}-\xmat{W})^{-1}\xmat{\Phi}_p =(\xmat{I}-\xmat{W})^{-1} \xmat{\Phi}_1 + (\xmat{I}-\xmat{W})^{-1}\xmat{\Phi}_2 + \ldots + (\xmat{I}-\xmat{W})^{-1}\xmat{\Phi}_P$$
Using the triangle inequality property of the operator norm, we have,
\begin{align*}
\left\| \sum_{p = 1}^{P} (\xmat{I}-\xmat{W})^{-1}\xmat{\Phi}_p \right\| &\leq \sum_{p = 1}^{P} \| (\xmat{I}-\xmat{W})^{-1}\xmat{\Phi}_p \| \\
&= \sum_{p = 1}^{P} \| (\xmat{I}-\xmat{W})^{-1} \|\cdot \| \xmat{\Phi}_p \| 
\end{align*}
If the elements of $\xmat{\Phi}_p$ are less than 1,  then the maximum norm of $\xmat{\Phi}_p$ is less than 1, denoted by $\| \xmat{\Phi}_p \|_{\text{max}} < 1$. For instance, if $\sum_{p=1}^{P}\phi_{ip} < 1$ for all $i = 1, ..., n$, i.e., the diagonal elements of $\xmat{\Phi}_p$ sum to a value smaller than one for each location, we get that $\sum_{p = 1}^{P} \| \xmat{\Phi}_p \| < 1$.

\noindent Similarly, for spatial stability/stationarity, assuming that $\xmat{W}$ is uniformly bounded in row and column sums in absolute value \cite{yu2008quasi},  then $||(\xmat{I}-\xmat{W})^{-1}|| < 1$ denoted by $\| (\xmat{I}-\xmat{W})^{-1} \|_{\text{max}} < 1$. Consequently, the series $(\xmat{I}-\xmat{W})^{-1} = \xmat{I} + \xmat{W} + \xmat{W}^2 + ...$ converges and the inverse $(\xmat{I}-\xmat{W})^{-1}$ exists. Thus, we have:
\begin{align*}
\left\| \sum_{p = 1}^{P} (\xmat{I}-\xmat{W})^{-1}\xmat{\Phi}_p \right\| &\leq \sum_{p = 1}^{P} \| (\xmat{I}-\xmat{W})^{-1} \|\cdot \| \xmat{\Phi}_p \|< 1\\
\end{align*}
and the process is stationary across space and time.

\end{proof}
\section{Parameter Estimation}
In this section, we discuss the penalised maximum likelihood technique in estimating the model's parameters. We also show how the penalties are introduced into the likelihood before optimisation. Further, we discuss the selection consistency of the LASSO estimators. 

\subsection{Penalised Maximum Likelihood Estimation}

We aim to estimate the model parameters, $\xvec{\Theta} = \{ \xvec{\beta}, \xmat{W}, \mathbf{\Phi}_1, \ldots, \mathbf{\Phi}_P, \sigma_\varepsilon^2\}$. Given the error term, $\xvec{\epsilon_t}$ expressed as the difference between the observed and predicted value of $\xvec{Y}_t$, we express it as 
$$\xvec{\varepsilon}_t = (\xmat{I}-\xmat{W})\xmat{Y}_t -\xmat{X}_t \xvec{\beta}-\sum_{p = 1}^{P} \mathbf{\Phi}_{p} \xvec{Y}_{t-p},$$ where the error term, $\xvec{\epsilon_t}$ is assumed to be independent and identically distributed (i.i.d.) with zero mean and constant variance $\sigma^2_\varepsilon$ i.e., $\xvec{\varepsilon}_t \sim \mathcal{N}(\xvec{0}, \sigma_\varepsilon^2\xmat{I}_n).$ 
Based on this assumption, the likelihood function at each time point $t$ will be expressed as;
$$f(\xvec{Y_t} | \xvec{\Theta}) = \frac{1}{(\sqrt{2 \pi \sigma_\varepsilon^2})^n} \exp \left( - \frac{\xvec{\varepsilon}_t^{\top}\xvec{\varepsilon}_t}{2 \sigma_\varepsilon^2} \right)|\text{det}(\xmat{I}_n-\xmat{W}_n)|$$
and thus the joint likelihood across $T$ time points is given by 
$$L(\xvec{\Theta}) = \prod_{t=1}^{T} \frac{1}{(2 \pi \sigma_\varepsilon^2)^{n/2}}\exp \left(- \frac{\xvec{\varepsilon}_t^\top \xvec{\varepsilon}_t}{2 \sigma_\varepsilon^2} \right)|\text{det}(\xmat{I}_n-\xmat{W}_n)|$$ and the corresponding log-likelihood function is given by:  
\begin{equation*}
\ln L(\xvec{\Theta}) = T \ln |\text{det}(\xmat{I}_n-\xmat{W}_n)|-\frac{nT}{2} \ln 2 \pi \sigma_\varepsilon^2 -\frac{1}{2\sigma_\varepsilon^2}\sum_{t = 1}^{T}\xvec{\varepsilon}_t^{\top}\xvec{\varepsilon}_t
\end{equation*} 
To enforce sparsity in the parameter estimates, we introduce $l_1-$penalties to the likelihood function and define the penalised negative log-likelihood as; 
$$-\ln L(\xvec{\Theta}) + \, \left( \lambda_1\sum_{i,j = 1}^{n} |w_{ij}| + \lambda_2\sum_{i = 1}^{n}\sum_{p = 1}^{P} |\phi_p(s_i)| + \lambda_3\sum_{k = 1}^{k} |\beta_k| \right)$$ 
where $\lambda = (\lambda_1,\lambda_2,\lambda_3) $ is the $l_1$ penalties.  Each set of parameters is subjected to an $l_1$ penalty due to its capacity to induce sparsity in parameter estimates \citep{tibshirani1996regression}. The paper employs three different LASSO penalties for the set of regressor coefficients, temporal coefficients, and weights. The $l_1$ penalty shrinks some of the $\beta$'s, $\phi$'s, and $w_{ij}$ to exactly zero if they have less or no significance in the model. This ensures precise control over model complexity, ensuring an adept capture of spatiotemporal dynamics while maintaining interpretative clarity and predictive efficacy of the spatiotemporal autoregressive model. More precisely, the estimates are obtained by the objective function with a constrained LASSO given by
\begin{align*}
\hat{\xvec{\Theta}} \;\; =  \;\; \arg\min_{\xvec{\Theta}} \; & - L(\xvec{\Theta}) +\left(\lambda_1 \sum_{i,j = 1}^{n} |w_{ij}| + \lambda_2\sum_{i = 1}^{n}\sum_{p = 1}^{P} |\phi_p(s_i)| + \lambda_3\sum_{k = 1}^{k} |\beta_l| \right) \\& \text{subject to: }  \\& w_{ij} \geq 0 , \, w_{ii} = 0 , \sum_{j=1}^{n} w_{ij} < 1 \; \text{for all} \; 1 \leq i,j, \leq n  \, ,
\\& ||\sum_{p = 1}^{P}(\xmat{I}-\xmat{W})^{-1}\xmat{\Phi}_p|| <  1 \; \; \text{if} \; \; 0\leq \phi_1(s_i) \leq 1\; \, \; \forall i. 
\end{align*}
All parameter estimates are obtained conditional on the first $P$ observations $\xvec{y}_1, \ldots, \xvec{y}_P$ since the model includes a temporal autoregressive component of order $P$. The constraints ensure that the model is well-defined in terms of the non-singularity of $(\xmat{I} - \xmat{W})$ and that the spatiotemporal process is stationary, 
\textit{see Proposition \ref{prop:stationarity}}. The constraint term $||\xmat{W}|| <1 $ could be computationally implemented by ensuring $w_{ij} \geq 0 , \, w_{ii} = 0 ,\sum_{j=1}^{n} w_{ij} < 1 \; \text{for all} \; 1 \leq i,j, \leq n $. We do not assume further restrictions on $\xmat{W}$ such as symmetry or directional dependency, thus its estimation is not necessarily unique \citep{kelejian1999generalized}. In particular, this is the case if there is no sufficient variation in the exogenous regressor $\xmat{X}_t$ which helps for the identification of $\xmat{W}$. It should be noted that the predicted spatial lag term $\hat{\xmat{W}}\xvec{Y}$ is always unique \citep{tibshirani2013lasso, ali2019generalized}. Therefore, the single estimated spatial weights should be interpreted with caution, but they will always indicate the true underlying spatial dependence structure. A review of identification of spatial spillovers in the econometric setting can be found in \cite{debarsy2025identification}.

\noindent The Sequential Quadratic Programming (SQP) algorithm is used in this optimisation by iteratively adjusting the parameters for minimising the objective function while incorporating constraints. The algorithm leverages the gradient and Hessian of the objective function to guide parameter updates, ensuring efficient convergence toward the optimal solution by considering both the direction and curvature of the objective function landscape \citep{wright2006numerical}. Though effective in its performance, its implementation in high-dimensional parameter spaces presents computational challenges, as the number of parameters increases with the number of spatial units. Additionally, the non-convexity of the objective function and the algorithm's sensitivity to initialisation can affect its performance, necessitating additional computational resources \citep{zhang2009some}.

\subsection{Selection Consistency of LASSO Estimators}

Selection consistency is a crucial property in the estimation of spatial weights, especially when utilising LASSO procedures. Each weight in this context represents a link between two observations, and identifying these links accurately is of significant interest in practical applications. However, it is debatable whether all individual links can be identified reliably. In particular, the challenge lies in the potential for missed links between two arbitrary locations to be compensated by introducing incorrect links via a third, unrelated location. For instance, a direct link from location A to B might be missed, but the model could incorrectly compensate by establishing links from A to C and C to B. Moreover, in the case of two-sided or asymmetric links, i.e., where the weight from A to B differs from the weight from B to A, it is particularly challenging to reveal those weights. Since any estimator would only have information based on the (sample) correlation of the observation, which is encoded in the symmetric covariance matrix, repeated observations across time or spatial subsets are needed \citep[cf.][]{otto2023estimation,merk2022estimation}. This problem is particularly relevant because the spatiotemporal setting is high-dimensional, i.e., the number of weights to be estimated may exceed the number of available observations. Specifically, while the number of weights is of the order $O(n^2)$, the number of observations is of the order $O(nT)$. Regressing each observation on all other observations, weighted by individual parameters, parallels the task of estimating the spatial weight matrix and was considered in the literature on covariance or neighbourhood selection \citep{dempster1972covariance,meinshausen2006high}.

LASSO, in general, provides a selection consistent estimate, as considered in \cite{zhao2006model,fan2001variable,zhang2008sparsity}. Furthermore, \cite{meinshausen2006high} have shown the selection consistency of a standard LASSO estimator for the neighbourhood selection. The neighbourhood $ne_A$ of a node/location A in a graph $G$ is defined as the smallest subset of $\Gamma \setminus \{a\}$ such that, given all variables $X_{ne_A}$ in the neighbourhood, $X_a$ is conditionally independent of all remaining variables. Assuming independent and homoscedastic normal errors, $\boldsymbol{\varepsilon}_t \sim N_n(\boldsymbol{0}, \sigma^2 \mathbf{I}_n)$, independently for all $t$, we can rewrite the model as
\[ \boldsymbol{Y}_t = (\mathbf{I}_n - \mathbf{W})^{-1} \left(\sum_{p=1}^{P}\mathbf{\Phi}_p Y_{t-p} + \mathbf{X}_t \boldsymbol{\beta} + \boldsymbol{\epsilon}_t\right) \, .\]
For simplicity, assume $P = 1$ and a weak temporal dependence structure\footnote{Note that, for a complete spatiotemporal solution, the covariance matrix needs to solve the Lyapunov equation considering all temporal components $\mathbf{\Phi}_p$.}, then we can approximate the covariance as
\[ \text{Cov}(\boldsymbol{Y}_t) \approx \sigma^2_\varepsilon(\mathbf{I} - \mathbf{W})^{-1}(\mathbf{I} - \mathbf{W}')^{-1} \]
and the precision matrix
\[ \Omega \approx \frac{1}{\sigma^2_\varepsilon} (\mathbf{I} - \mathbf{W}')(\mathbf{I} - \mathbf{W}) \, . \]
Thus, the conditional independence between observations is directly encoded in the weight structure given by $\mathbf{W}$. If the underlying dependence structure is sparse, \cite{meinshausen2006high} proved the neighbourhood selection consistency when regressing each observation with all others, effectively controlling both type-1 and type-2 errors.

\section{Simulation Study}
A Monte Carlo simulation was set up to assess the efficiency of the model to detect existing spatial and temporal dependencies (connections and relationships between various spatial locations/units) within the data. Precisely, the simulation assesses how LASSO penalties can distinguish between meaningful spatial weights (non-zero weights) from insignificant weights (zero weights).  

\subsection{Spatiotemporal design}
The spatiotemporal process is simulated on a spatial lattice of $ n \in \{4, 9, 16,  25\}$ spatial units/locations, discrete time points $t = 1,\ldots,T$ with $T \in \{50, 100 , 200\}$ and $ m = 100$ replications. The temporal dependency was set up such that the first $25\%$ of the spatial units had no temporal dependency and the remaining spatial units had a uniform temporal dependency of $0.3$ with a time lag of $P = 1$. A row-standardised weighting scheme based on Queen’s contiguity criterion and then multiplied with a moderate spatial coefficient, $\rho=0.6$, was considered when defining the spatial dependency. A three-dimensional array of exogenous regressors was generated with dimensions corresponding to the number of covariates ($k$) with coefficients, $\xvec{\beta}=(3, 0, 2)'$, spatial units ($n$), and time points ($T$). Specifically, these regressors were simulated from a normal distribution, ensuring that each covariate is independent across spatial units and time points and has a unit variance. The errors generated are independent and identically distributed from a normal distribution with $\sigma_\varepsilon^2 = 1$ to ensure the model is evaluated, validated, and assessed. The resulting number of parameters to be estimated is $k+n^2+1$. 

To determine the optimal regularisation parameters for the LASSO model, we employed block cross-validation, preserving both spatial and temporal dependencies within the data \citep[see][for a review of cross-validation methods in spatiotemporal statistics]{otto2024review}. The dataset was partitioned into five blocks, each consisting of $n \times \frac{T}{5}$ observations, where $n$ represents the number of spatial locations and $T$ is the total number of periods. The blocks were sequentially constructed, with the first block containing data from time $i = 1$ to $\frac{T}{5}$, the second block from $\frac{T}{5} + 1$ to $2 \times \frac{T}{5}$, and so on. This method allowed for the independent evaluation of model performance while ensuring that the spatial and temporal structures of the data were preserved. The LASSO penalty combination that minimised the root mean squared error (RMSE) across the blocks was chosen as the optimal regularisation for the full optimisation.

\subsection{Simulation results}
The simulation results obtained from estimating the spatial weights matrix for the different spatiotemporal settings are shown in Figure \ref{weights}. For each spatial unit, it is observed that the model improves in its estimation capability as the number of time points increases. With an increase in the number of spatial units, for instance, $n=9$ spatial units, the model performs better in estimating the weights, with the prediction accuracy increasing with an increase in time points. 
\begin{figure}
\centering
\includegraphics[width=17cm]{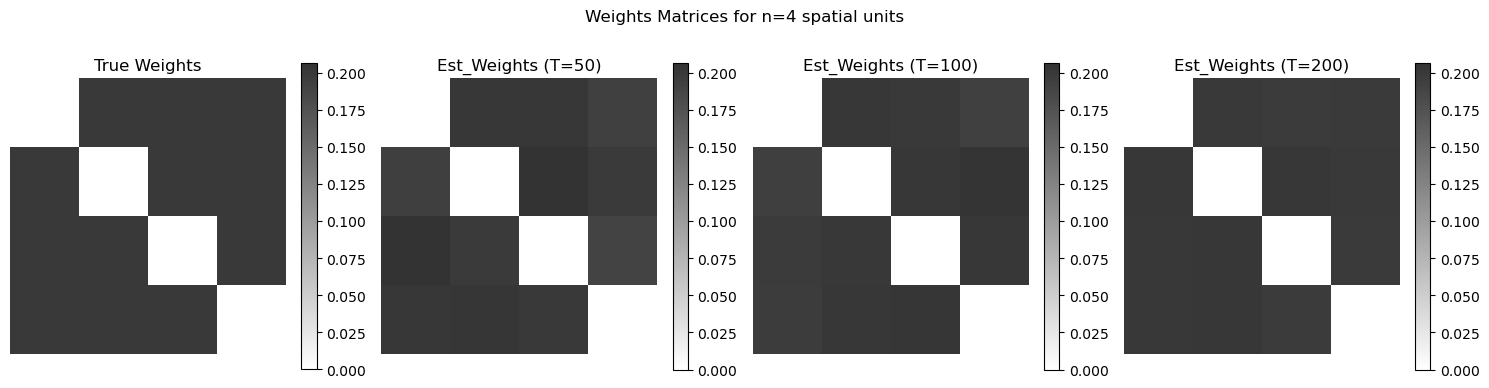}
\includegraphics[width=17cm]{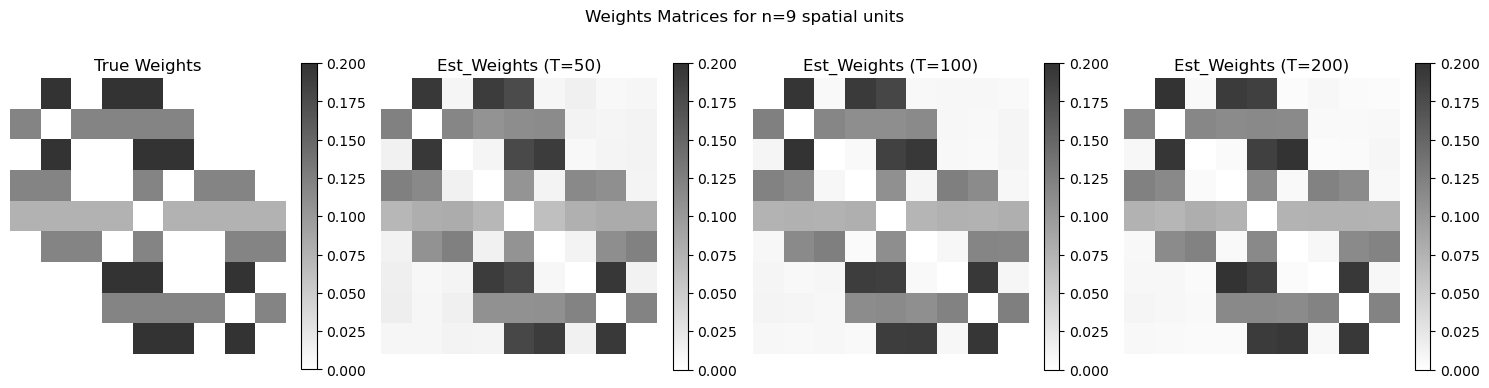}
\includegraphics[width=17cm]{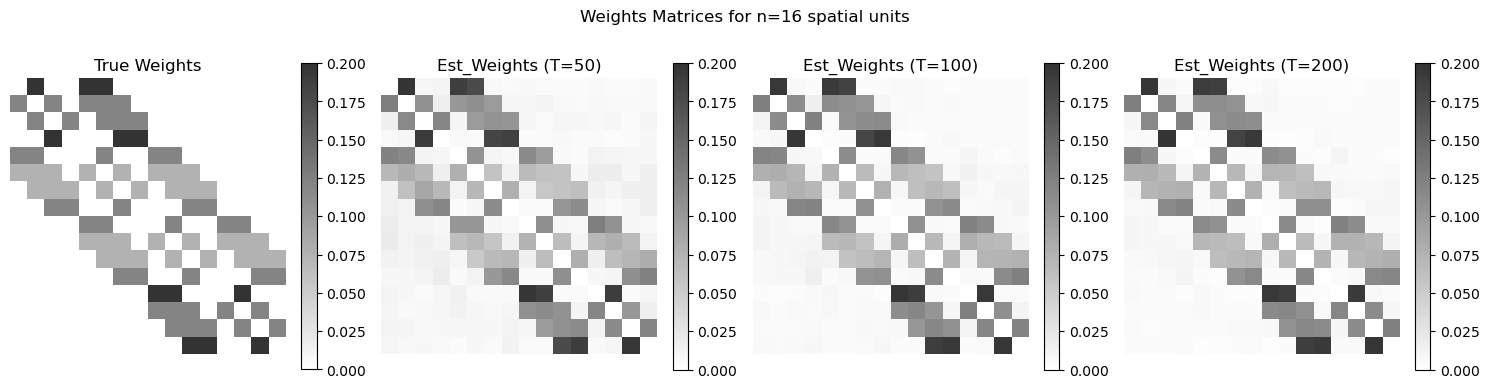}
\includegraphics[width=17cm]{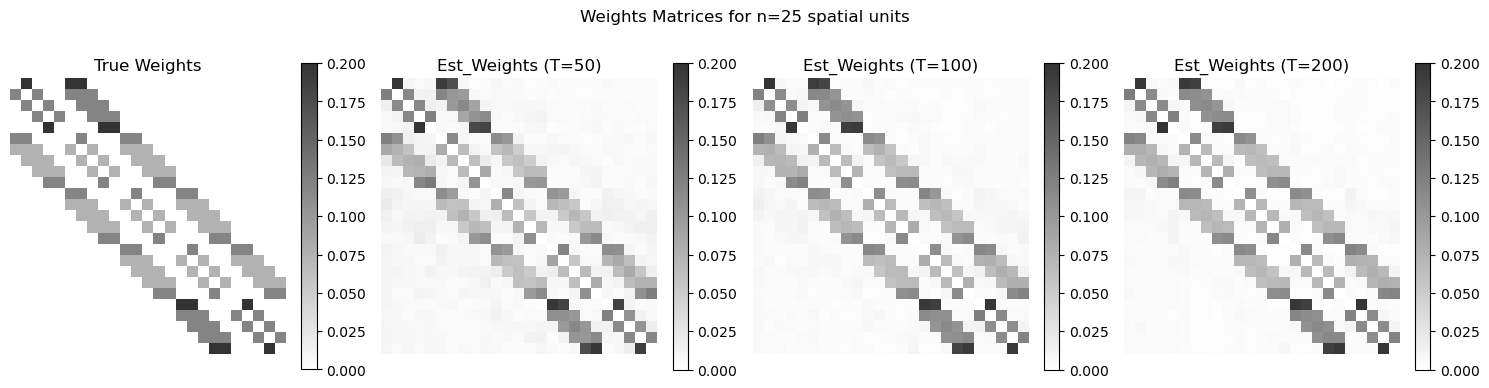}
\caption{Comparison of true weights (left) and estimated weights (right) for four, nine, sixteen, and twenty-five spatial units across 50, 100, and 200 time-points respectively.}
\label{weights}
\end{figure}
\noindent Recall that the total number of parameters estimated given by $n^2 + k + 1$ is quadratically increasing in the number of locations, resulting in a high demand for computational resources. Assessment of the model performance of the spatiotemporal model framework is conducted for each set of parameters by considering the average bias, average relative bias, and the root mean squared error at different time points as,
\begin{eqnarray*}
\overline{\text{Bias}}_\beta & = & \frac{1}{mk}\sum_{i=1}^m(\hat{\xvec{\beta}}_{i}-\xvec{\beta})'\xvec{1},\\
\overline{\text{Bias}}_\phi & = &\frac{1}{nm}\sum_{i=1}^m(\hat{\xvec{\phi}}_{i}-\xvec{\phi})'\xvec{1},  \\
\overline{\text{Bias}}_w & = & \frac{1}{m(n^2-n)}\sum_{i=1}^m (\hat{\xvec{w}}_{i}-\xvec{w})'\xvec{1} \; , \;\; \text{and} \\
\overline{\text{Bias}}_\sigma & = & \frac{1}{m}\sum_{i=1}^m(\hat{\sigma}_i^2-\sigma^2_\varepsilon)
\end{eqnarray*}
for the regressor coefficients, $\xvec{\beta}$'s, the temporal coefficients, $\xvec{\phi}$'s, spatial weights, $\xvec{w}$'s and error variances, $\sigma^2_\varepsilon$, respectively. Note that $\xvec{\phi}$ is the vector of diagonal elements of the temporal coefficients matrix while $\xvec{w}$ is the vectorised spatial weights matrix excluding the diagonal elements which are zero by definition. The root mean square error (RMSE) of the fitted process parameters is computed as follows: 
\begin{eqnarray*}
RMSE_\beta &= & \sqrt{\frac{1}{mk}\sum_{i=1}^m(\hat{\xvec{\beta}}_{i}-\xvec{\beta})^{'}(\hat{\xvec{\beta}}_{i}-\xvec{\beta})},\;\; \\
RMSE_\phi & = & \sqrt{\frac{1}{nm}\sum_{i=1}^m(\hat{\xvec{\phi}}_{i}-\xvec{\phi})^{'}(\hat{\xvec{\phi}}_{i}-\xvec{\phi})},  \\
RMSE_w & = & \sqrt{\frac{1}{m(n^2-n)}\sum_{i=1}^m (\hat{\xvec{w}}_{i}-\xvec{w})^{'}(\hat{\xvec{w}}_{i}-\xvec{w})} \; , \;\; \text{and} \;\; \\
RMSE_\sigma & = & \sqrt{\frac{1}{m}\sum_{i=1}^m(\hat{\sigma}_i^2-\sigma^2)^2}.
\end{eqnarray*}
Moreover, the full model root mean squared error is given by
$$RMSE = \sqrt{\frac{1}{nm}\sum_i^m (\hat{\xvec{Y}}-\xvec{Y}_{i})^{'}(\hat{\xvec{Y}}-\xvec{Y}_{i})}.$$ The average time taken for each of the estimations across replications is also recorded in hours alongside the simulation results as shown in Table \ref{tab:results}.

\begin{sidewaystable}[!p]
\begin{small}
\begin{tabular}{cccccccccccccc}
\hline
\multirow{3}{*}{Parameter} & \multirow{3}{*}{Measure} & \multicolumn{12}{c}{No of Spatial units \& Time points} \\\cline{3-14}
                           &                        & \multicolumn{3}{c}{$n = 4 \; \;(\text{par}=20)$}  & \multicolumn{3}{c}{$n = 9 \; \;(\text{par}=85)$} & \multicolumn{3}{c}{$n = 16 \; \;(\text{par}=260)$}  & \multicolumn{3}{c}{$n = 25 \; \;(\text{par}=629)$}  \\\cline{3-14}
                           &                          & $T = 50$     &  $100$  & $200 $  & $50 $    &  $100 $ & $200$ & $50$ &  $100$  & $200$   & $50$     &  $100$  & $200  $ \\
\hline
\multirow{2}{*}{Beta ($\xvec{\beta}$)}      & Bias    & 0.0032  & -0.0015   & 0.0058     & 0.0006   & 0.0018      & 0.0260 & 0.0093 &  0.0041  & 0.0037   & 0.0042  & 0.0080    & 0.0048   \\
&  MAE   & 0.0660  & 0.0446  &   0.0309    &  0.0433  & 0.0303      & 0.0189 &0.0321 &  0.0217  & 0.0146  & 0.0270 &  0.0182   &  0.0137 \\
& RMSE    & 0.0834  & 0.0547 & 0.0378      & 0.0537  & 0.0370     & 0.0239&  0.0400  & 0.0273  &  0.0183      & 0.0340  &  0.0239     &  0.0168  \\
\hline
\multirow{2}{*}{Phi ($\xmat{\Phi}$)}        & Bias    & 0.0015  & 0.0013  & 0.0009       & 0.0005   & 0.0003    & 0.0003 & 0.0014  & 0.0008 &  0.0029   & 0.0011  & 0.0017     & 0.0048    \\
& MAE    & 0.0185  &  0.0131 &   0.0093    &  0.0229  &  0.0157     & 0.0108 &0.0251 &  0.0176  & 0.0129  & 0.0262 &  0.0182   &  0.0138 \\
& RMSE    & 0.0247  & 0.0177   & 0.0122       & 0.0274   & 0.0199     & 0.0138    & 0.0324 & 0.0224    & 0.0165 & 0.0340  & 0.0234     & 0.0176   \\\hline
\multirow{2}{*}{$\xmat{W}$}           & Bias          & -0.0006  & -0.0004   & 0.0000      &0.0012    & 0.0011     & 0.0008 & 0.0025  & 0.0018  &  0.0014   & 0.0340  & 0.0023     & 0.0019    \\
&  MAE    &  0.0319 & 0.0226  &   0.0146    &  0.0232  &  0.0157     & 0.0112 & 0.0177   & 0.0121  & 0.0088 &   0.0142  &  0.0096 & 0.0073  \\
& RMSE    & 0.0406  & 0.0283   & 0.0183      & 0.0325    & 0.0222    & 0.0154 & 0.0267  & 0.0182  &  0.0135   & 0.0232  & 0.0157     & 0.0123   \\\hline
\multirow{2}{*}{$\xmat{W} = 0$}           & Bias      & --      &   --     &    --  & 0.0112   & 0.0076     & 0.0062 & 0.0088  & 0.0062  &  0.0050  & 0.0077  & 0.0054     & 0.0044  \\
&  MAE    & --  & --  &  --     &  0.0112  &   0.0076    & 0.0062 & 0.0088 &  0.0062  & 0.0050  & 0.0077 &  0.0054   & 0.0044  \\
& RMSE    & --      &   --     &    --       & 0.0220    & 0.0150    & 0.0109 & 0.0183  & 0.0125  &  0.0101   & 0.0169  & 0.0116     & 0.0098    \\\hline
\multirow{2}{*}{$\xmat{W} \neq 0$ }        & Bias     & -0.0006  & -0.0004    & 0.0000     & -0.0068    &-0.0042    & -0.0036 & -0.0093  & -0.0064  &  -0.0054  & -0.0103  & -0.0073     & -0.0058   \\
& MAE    &  0.0319 & 0.0226  &  0.0146     &  0.0223  &  0.0152     & 0.0112 & 0.0341& 0,0230   & 0.0158  & 0.0348 &  0.0229   &  0.0165 \\
& RMSE     & 0.0406  & 0.0283    & 0.0183     & 0.0409   & 0.0279    & 0.0190 & 0.0423  & 0.0287  &  0.0197   & 0.0430  & 0.0286    & 0.0205    \\\hline
\multirow{2}{*}{Sigma $\sigma_\varepsilon^2$}    & Bias  & -0.0966  & -0.034    & -0.0171      & -0.1361  & -0.0698     & -0.0286 & -0.1762  & -0.0847  &  -0.0375   & -0.2248  & -0.1045     & -0.0418    \\
&  MAE    & 0.1199  & 0.0571  &  0.0402     &  0.1363  &  0.0742     & 0.0344 & 0.1762 & 0.0847   & 0.0392  & 0.2248 & 0.1045    & 0.0423  \\
& RMSE  & 0.1429  & 0.0683   & 0.0508      & 0.1492    & 0.0871    & 0.0435 & 0.1823  & 0.0928  &  0.0448   & 0.2275  & 0.1079     & 0.0462  \\\hline
Time $t$   &  hours  & 1.1454  & 1.2235   & 1.5302    & 0.3712  & 1.0201   & 4.4065  & 2.6511  & 3.9110  &  6.5554  &  13.4742 &  16.8284   & 26.3437  \\\hline                     
\end{tabular}
\end{small}
\caption{Mean Bias and Root Mean Square Error (RMSE) of Estimated Parameter Values}
\label{tab:results}
\end{sidewaystable}
\noindent In our analysis, we examine the model estimation performance of the entire weights matrix $\xmat{W}$, focusing on both non-zero and zero elements separately. This implies that the row with $\xmat{W}$ in Table \ref{tab:results} has all elements of the weights both zeros and non zeros, row with $\xmat{W}= 0$ considers elements of $\xmat{W}$ that are zero only while the row $\xmat{W} \neq 0$ considers non-zero elements of the spatial weights matrix $\xmat{W}.$ The mean absolute error is also calculated for each case to assess the accuracy of parameter estimation relative to the true values. For the whole weights matrix, the bias provides an overview of the overall accuracy of the estimation process. In contrast, analysing only the non-zero elements evaluates how accurately the model captures spatial relationships represented by these weights. Similarly, the estimation performance of zero elements in $\xmat{W}$ is assessed to understand how well the model identifies the absence of spatial relationships among certain units. These analyses offer insights into the effectiveness of the spatiotemporal modelling framework in capturing spatial dependencies within the data.

\noindent In general, the estimation accuracy of all parameters are increasing with an increase in time points while the average bias decreases with an increase in time points, $T$. For the case of weights, this signifies that the estimation of the spatial dependence structure in terms of the magnitude of the weights and the number of correctly estimated zero and non-zero weights was getting more precise when the number of time points increased. One captivating observation is the consistent improvement of estimation accuracy with an increase in time points, particularly with regard to the spatial dependency structure. A richer temporal dataset will provide a key advantage in uncovering the intricacies of spatial dependencies and ensuring the precision of the model's outputs, thereby enhancing its applicability to real-world scenarios. 

\section{Spatiotemporal air quality modelling}

Globally, air pollution is still a significant cause for concern, primarily due to its harmful impacts on both human well-being and the environment. According to \cite{EEA, EPA_PM}, high levels of particulate matter (PM) concentrations are concerning as they contribute significantly to respiratory and cardiovascular diseases hence need for monitoring and control strategies. 

\noindent Statistical models have been widely adopted to address this by detecting pollution spatial and temporal patterns, predicting PM concentrations exceedances and assessing the influential factors affecting their dynamics \citep{fasso2022spatiotemporal, otto2024spatiotemporal, krylova2022managing}. Building on these methods, our spatiotemporal autoregressive model aims at uncovering the spatial patterns and temporal trends in air quality monitoring stations in Bavaria, Germany. The hourly PM concentration data spanning from 2005 to 2020 was obtained from the German Environmental Agency (Umweltbundesamt, UBA).
\begin{figure}
\centering
\includegraphics[width=15cm]{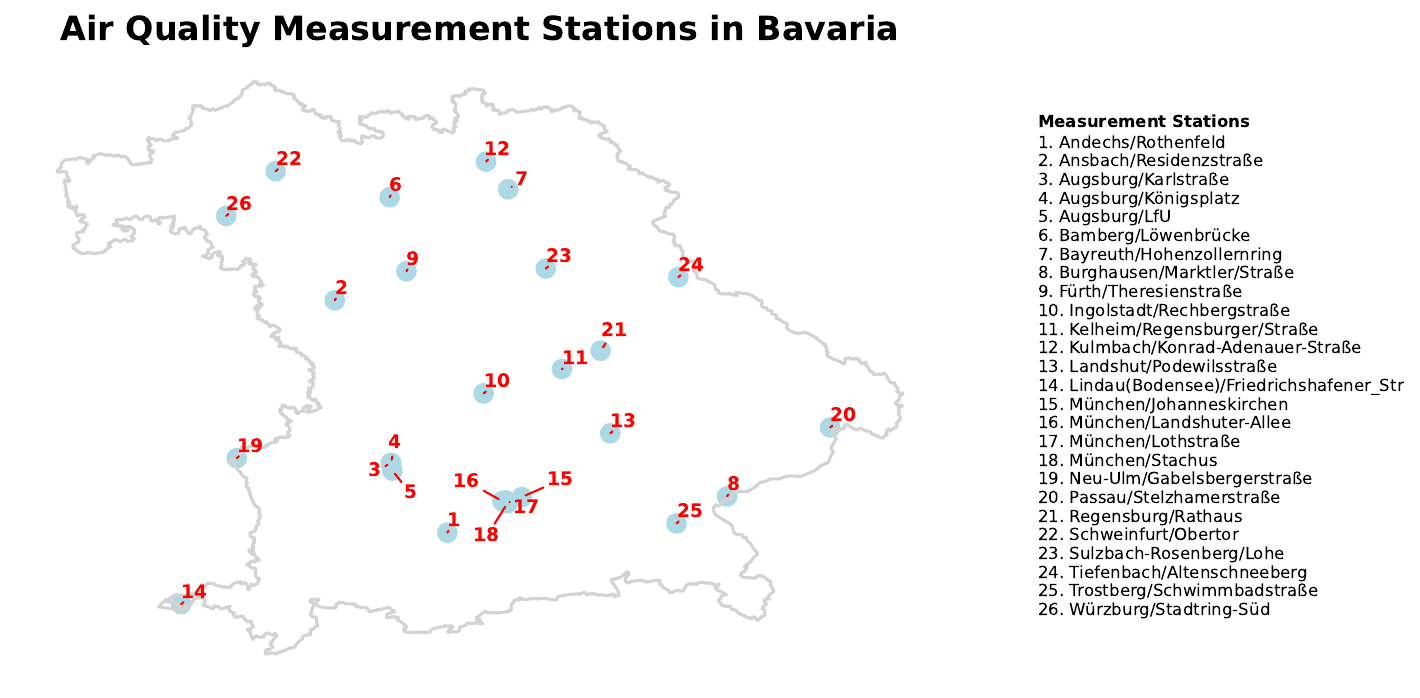}
\caption{Air quality measurement stations in Bavaria, Germany (some stations are in the same geographical zone e.g., see stations 3, 4 and 5, so the label was plotted outside.)}
\label{Data source}
\end{figure}
Data from 26 stations with over $90\%$ of measurements recorded was used, with missing values imputed using the backward-forward fill method. The average PM concentration across the stations was 21.80 $\mu g/m^3$, with a standard deviation of 18.52 $\mu g/m^3$ and a median of 18.00 $\mu g/m^3$. Despite a decreasing trend over the years of the time series (see Figure \ref{fig: PM_concetrations}), the spatial dependence of these concentrations is still an issue. The time series data shows a decreasing trend in particulate matter concentrations over the years with a significant seasonal structure (see Figure \ref{fig: predicted_residuals}). 
\begin{figure} 
    \centering
     \includegraphics[width=\linewidth]{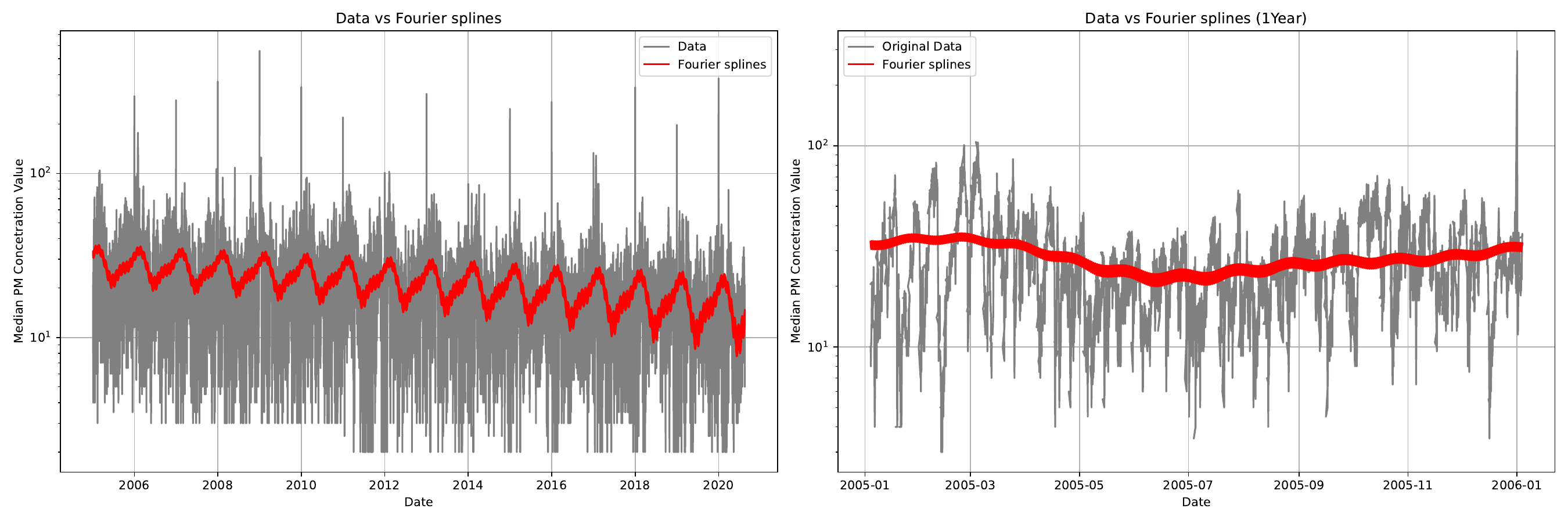}
    \caption{Figure showing median PM Concentration values and fourier splines across monitoring stations}
    \label{fig: PM_concetrations}
\end{figure}

\noindent We model these dependencies together with other parameters by using the spatiotemporal autoregressive model 
\begin{equation*}
\xvec{Y}_{t} = \xmat{X}_{t} \xvec{\beta} + \Phi_{1} \xvec{Y}_{t-1} + \xmat{W}\xvec{Y}_{t} + \xvec{\varepsilon}_{t},
\end{equation*}
and adding LASSO penalties to the estimation procedure. The model consists of $\xvec{Y}_{t}$, the level of PM concentration at a given hour, $t$, $\xmat{X}_t$ representing the design matrix that incorporates Fourier components of daily, monthly, biannually, and yearly frequencies. This captures the yearly cycles and seasonal variations in the PM concentrations over time across all monitoring stations and $\xvec{\beta}$'s are their coefficients. Moreover, $\xvec{Y}_{t-1}$ represents PM concentration levels in the previous hour and $\xmat{W}$ denotes the unknown spatial dependency for the air quality measurement stations in this study. The model also includes temporal coefficient for each measurement station, $\mathbf{\Phi}_{1}=\text{diag}(\phi_{1}(\xvec{s}_1), \ldots, \phi_{1}(\xvec{s}_n))$. We have used a temporal lag order of $P = 1$, which appears to be sufficient for our data, but generally, a suitable lag order could also be obtained by cross-validation across different candidate orders (e.g., identified by looking at the empirical autocorrelation functions (ACF)). The predicted PM concentrations show a reasonable fit, as displayed in Figure \ref{fig: predicted_residuals}  along with the residuals.
\begin{figure}
\centering
\includegraphics[width=17cm]{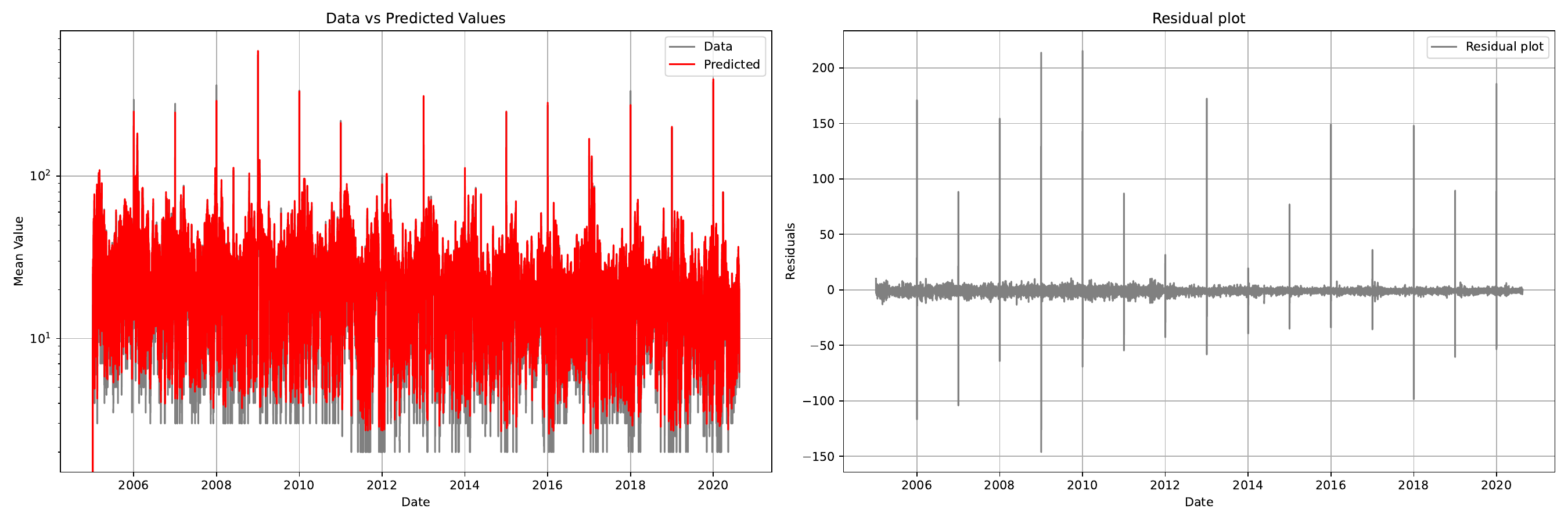}
\caption{Left: Graph of predicted data versus original data for the monitoring stations. Right: Plot of residuals}
\label{fig: predicted_residuals}
\end{figure}

\noindent  To do statistical inference, we re-estimate the model without penalisation, but only using the non-zero weights and covariates selected by the LASSO estimation. Then, the standard errors of these estimates are obtained empirically by computing the Hessian matrix of the non-penalised log-likelihood function evaluated at the estimated parameters. Specifically, the Hessian matrix $ H(\theta)$ is the matrix of second-order partial derivatives
$$H(\theta) = \frac{\partial^2 \ell(\theta) }{\partial \theta \, \partial \theta^\top},$$
where $\ell(\theta)$ is the non-penalized log-likelihood function and \( \theta \) is the vector of estimated parameters. Due to the complex structure of the model, we compute the Hessian \textit{numerically} using a central finite difference approximation. The covariance matrix of the parameter estimates is approximated by obtaining the inverse of the Hessian evaluated at the estimated parameters  $\xvec{\hat{\theta}}$:
$$\text{Cov}(\hat{\theta}) \approx H(\hat{\theta})^{-1}.$$

\noindent The standard errors for each estimated parameter $\hat{\theta}_i $ are then obtained from the diagonal elements of the covariance matrix:
$$
\text{SE}(\hat{\theta}_i) = \sqrt{\left( \text{Cov}(\hat{\theta}) \right)_{ii}}.$$

\noindent Moreover, the $ z $ values are computed as $z_i = \frac{\hat{\theta}_i}{\text{SE}(\hat{\theta}_i)}$ and the results reported in Table \ref{tab: Fourier estimates}.
\begin{table}[h!]
\centering
\caption{Fourier Parameter Estimates with Confidence Intervals}
\begin{tabular}{cccccc}\hline
\text{Parameter} & \text{Estimate} & \text{Standard Error} & \text{Z-Value} & \text{LCL} & \text{UCL} \\\hline
$\beta_1$ & -5.27496 & 0.030995 & -170.186 & -5.33531 & -5.21461 \\
$\beta_2$ & 2.87438  & 0.042798 & 67.162   & 2.79050  & 2.95827  \\
$\beta_3$ & -0.31387 & 0.022176 & -14.154  & -0.35733 & -0.27040 \\
$\beta_3$ & -0.19146 & 0.022525 & -8.500   & -0.23563 & -0.14729 \\
$\beta_4$ & 0.26611  & 0.022396 & 11.882   & 0.22222  & 0.31000  \\
$\beta_6$ & 0.09417  & 0.022565 & 4.173    & 0.04994  & 0.13840  \\
$\beta_7$ & -0.53290 & 0.023269 & -22.902  & -0.57851 & -0.48729 \\
$\beta_8$ & 0.67505  & 0.022316 & 30.250   & 0.63131  & 0.71880  \\
$\beta_9$ & -0.01737 & 0.022383 & -0.776   & -0.06125 & 0.02651  \\
$\beta_{10}$ & -0.26598 & 0.023455 & -11.340  & -0.31195 & -0.22001 \\\hline
\end{tabular}
\label{tab: Fourier estimates}
\end{table}

\noindent Furthermore, we have analysed the temporal autocorrelation of the residuals for each measurement station (see Figure \ref{fig: ACF station 26} and Figure \ref{fig:Autocorrelation plot} in the Appendix). The autocorrelation function (ACF) plots for the 26 air quality measurement stations depict a decay pattern in autocorrelation over increasing spatial distances. The ACF plots for the stations visually exhibit a random or near-zero pattern. We observe slight peaks of the ACF around 24 hours indicating a weak autoregressive structure. 
\begin{figure}
\centering
\includegraphics[width=17cm]{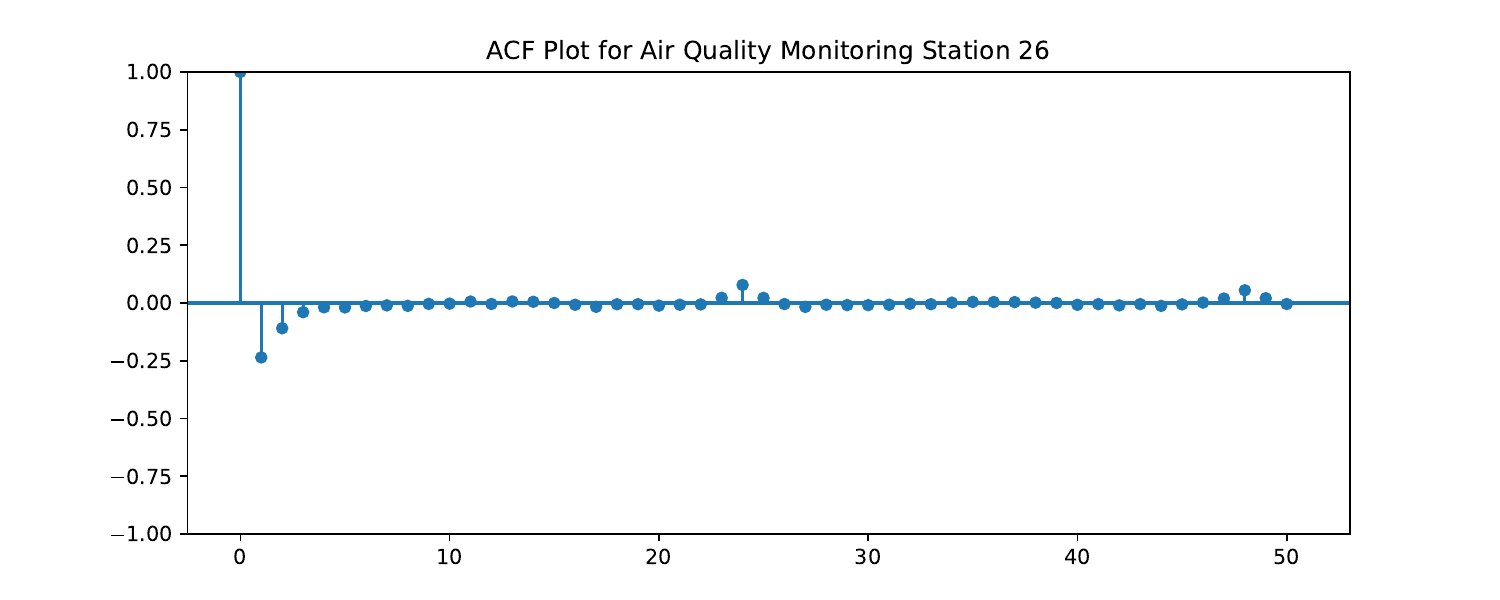}
\caption{Autocorrelation plot for ``Würzburg/Stadtring-Süd'' monitoring station}
\label{fig: ACF station 26}
\end{figure} 
Data analysis reveals high pollution peaks at the end of each year, likely linked to the widespread use of fireworks during New Year's celebrations, which temporarily increase pollution levels during the period. In future research, we could explore adding seasonal autoregressive structure. The model was then validated by comparing its Mean Squared Error (MSE), Akaike Information Criterion (AIC), and Bayesian Information Criterion (BIC) metrics with an Ordinary Least Squares regression approach (OLS), i.e., a model completely neglecting spatiotemporal dependence, and a Vector Auto-regression (VAR) model of order 1, which neglects the spatial dependence implied by $\xmat{W}$, and it proved to be more powerful with the lowest value for the three metrics as summarised in Table \ref{tab: tabcomparison}.  

\begin{table}
\centering
\begin{tabular}{cccc}\hline
\multicolumn{4}{c}{Model performance comparison} \\\cline{1-4}
 Measure & OLS & VAR  & Spatiotemporal model\\\hline
 MSE & 339.3519 & 327.7578  & 70.7844 \\
AIC & 30850000 & 20621861 & 15165678 \\
 BIC & 30850000 & 20630850 & 15174668 \\\hline
\end{tabular}
 \caption{Performance comparison of the spatiotemporal model with other models}
 \label{tab: tabcomparison}
 \end{table}

\subsection{Temporal dynamics of air quality in Bavaria}
The temporal dependencies over time for each station is estimated taking into consideration a one time point lag i.e., $\xmat{\Phi}_1 = \text{diag}(\phi_1(\xvec{s}_1), \ldots, \phi_1(\xvec{s}_n))$ for each of the air quality measurement stations. 

\begin{figure}
\centering
\includegraphics[width=17cm]{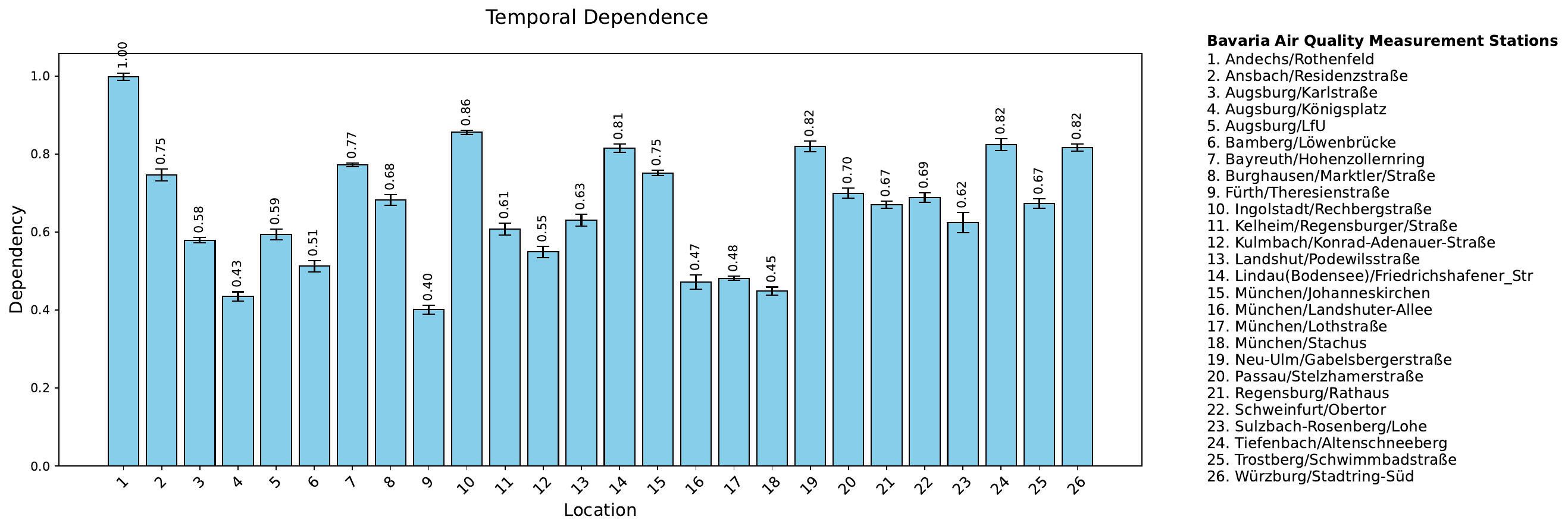}
\caption{Temporal dependency estimates with ±1.96 standard error bars (approximate $95\% $ confidence intervals) for each of the 26 air quality monitoring stations in Bavaria, based on a lag-one autoregressive model. }
\label{fig: temporal dependence}
\end{figure}

\noindent The temporal relationships between particulate matter (PM) concentrations across different air quality monitoring stations in Bavaria are not uniform, as demonstrated by the varying temporal coefficients observed in Figure \ref{fig: temporal dependence}. The bars represent the estimated strength of temporal autocorrelation at each station, with uncertainty derived from the Hessian of the log-likelihood function. This variability provides critical insights into the temporal dynamics of air quality in the region, with significant implications for environmental monitoring and policy-making. Notably, the ``M\"{u}nchen/Landshuter Allee" station exhibits a relatively moderate temporal coefficient of 0.47, implying moderate dependence of current PM concentration values on previous readings. This could suggest a stronger influence of short-term factors such as traffic or local industrial activities.

\noindent In contrast, the ``Andechs/Rothenfeld" and ``Ingolstadt/Rechbergstraße" monitoring stations' PM concentration values are highly influenced by previous concentration values at those stations with a coefficient of 0.99 and 0.86, respectively. Such patterns suggest persistent pollution sources or consistent environmental conditions, leading to more stable PM concentrations over time. The observed temporal dependencies highlight the importance of historical data in understanding and predicting PM concentrations, underscoring the need for continuous and comprehensive air quality monitoring. Given a decreasing trend over the years, we can observe that the PM concentration levels in the atmosphere are reducing significantly and, thus, locations with high air pollution need to benchmark and emulate what the state in Bavaria is implementing in mitigation efforts. 

\subsection{Spatial dependence across air quality monitoring stations}
The spatial dependencies between PM concentrations at different air monitoring stations show the distribution and variability of air pollution across Bavaria. As shown in Figure \ref{spatial and network}, the spatial relationships among the 26 monitoring stations show areas of both high and low correlation, underscoring the spatial heterogeneity of air pollution. Moreover, we summarise the spatial and temporal dependencies for each station in Table \ref{tab:locations}
\begin{table}[htbp]
    \centering
    \small
    \label{tab:locations}
    \scalebox{0.9}{
    \begin{tabular}{llllll}\hline
    Monitorng station & Latitude & Longitude & T.dependence & Sp.dependence & Location type \\\hline
    Andechs/Rothenfeld & 47.968753 & 11.220172 & High & Low & Rural \\
    Ansbach/Residenzstraße & 49.304889 & 10.572297 & High & Low & Urban traffic \\
    Augsburg/Karlstraße & 48.370237 & 10.896221 & High & Moderate & Urban traffic\\
    Augsburg/Königsplatz & 48.364590 & 10.895030 & Moderate & Moderate & Urban traffic\\
    Augsburg/LfU & 48.326013 & 10.903051 & High & Low & Sub urban \\
    Bamberg/Löwenbrücke & 49.898330 & 10.887686 & Moderate & Moderate & Urban\\
    Bayreuth/Hohenzollernring & 49.943637 & 11.570088 & High & Low & Urban \\
    Burghausen/Marktler/Straße & 48.177175 & 12.829314 & High & Low   & Sub urban\\
    Fürth/Theresienstraße & 49.472210 & 10.984705 & Moderate & Moderate & Urban traffic\\
    Ingolstadt/Rechbergstraße & 48.769444 & 11.428766 & High & Low & Urban traffic\\
    Kelheim/Regensburger/Straße & 48.909518 & 11.879245 & High & Low  & Urban traffic\\
    Kulmbach/Konrad-Adenauer& 50.103135 & 11.442591 & Moderate & Moderate & Urban \\
    Landshut/Podewilsstraße & 48.539801 & 12.157051 & High & Low  & Urban traffic\\
    Lindau(Bodensee) & 47.557473 & 9.688869 & High & Low  & Urban traffic\\
    München/Johanneskirchen & 48.17319 & 11.64804 & High & Low  & Sub urban\\
    München/Landshuter-Allee & 48.149606 & 11.536513 & Low & High & Urban traffic\\
    München/Lothstraße & 48.154534 & 11.554669 & Moderate & Moderate & Urban\\
    München/Stachus & 48.13732 & 11.56481 & Moderate & High & Urban traffic \\
    Neu-Ulm/Gabelsbergerstraße & 48.397079 & 10.008293 & High & Low  & Urban\\
    Passau/Stelzhamerstraße & 48.573629 & 13.422039 & High & Low  & Urban\\
    Regensburg/Rathaus & 49.015245 & 12.101557 & High & Low  & Urban traffic \\
    Schweinfurt/Obertor & 50.048397 & 10.232076 & High & Low  & Urban \\
    Sulzbach-Rosenberg/Lohe & 49.487983 & 11.786341 & High & Low  & Sub urban\\
    Tiefenbach/Altenschneeberg & 49.438464 & 12.548869 & High & Low  & Rural\\
    Trostberg/Schwimmbadstraße & 48.021657 & 12.538174 & High & Low  & Sub urban\\
    Würzburg/Stadtring-Süd & 49.790526 & 9.947654 & High & Low  & Urban traffic\\\hline
    \end{tabular}}
     \caption{Air quality monitoring stations characteristics data}
\end{table}

\begin{figure} 
\centering
\includegraphics[width=\textwidth]{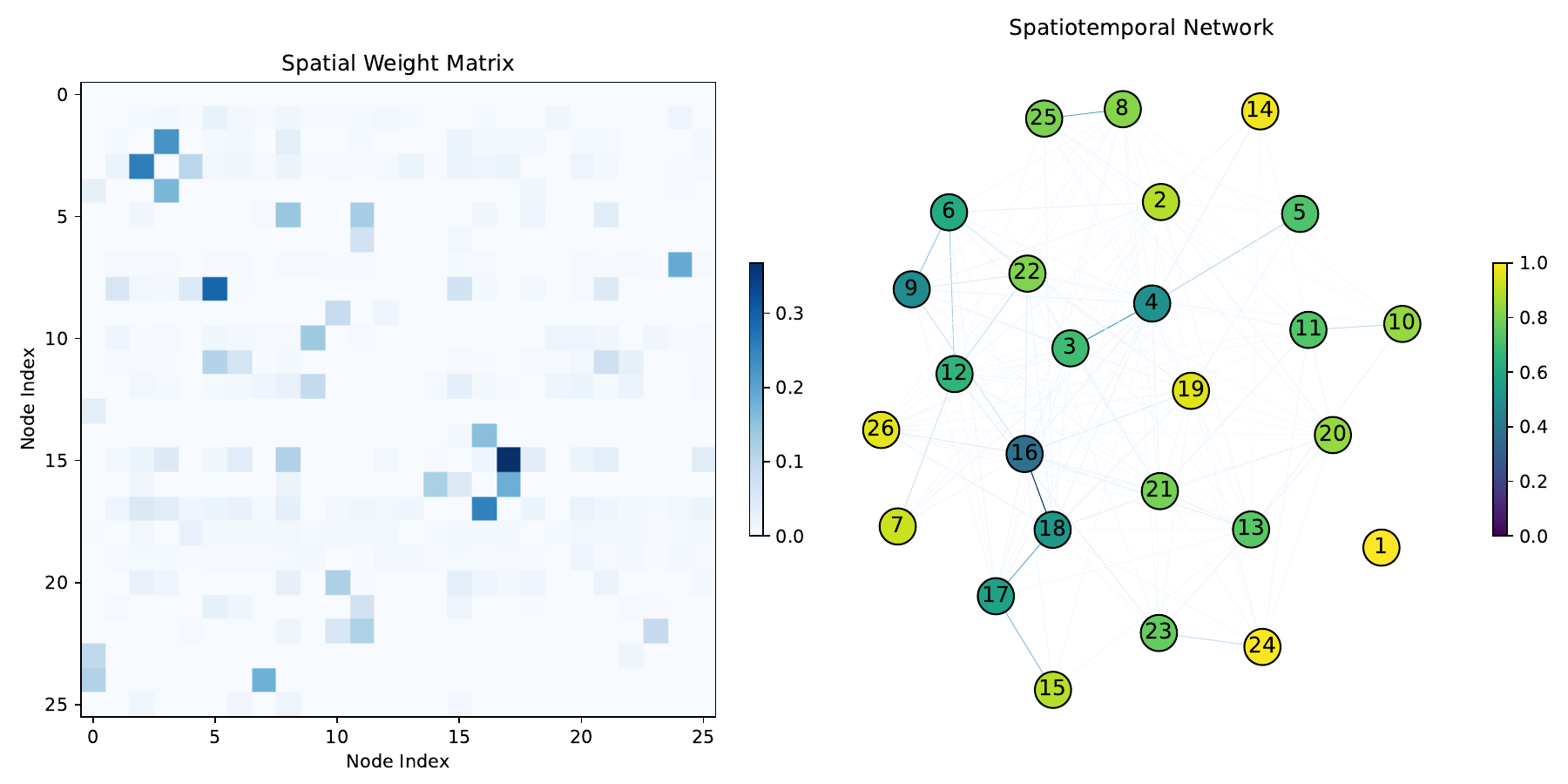}
\caption{ Estimated spatial dependence across monitoring stations. The left panel displays the estimated spatial weights matrix as a heatmap, where each cell represents the strength of spatial interaction between one station (row) and another (column). Darker shades indicate stronger dependencies. The right panel visualises the same spatial relationships as a network graph, where nodes represent monitoring stations and edges represent nonzero estimated spatial connections. Node colours indicate the strength of the corresponding station’s temporal dependency, with darker nodes showing higher temporal autocorrelation.} 
\label{spatial and network}
\end{figure}
\noindent When the spatial dependency is denoted as ``higher,'' it implies a strong correlation or relationship between the air quality readings of these air quality monitoring stations. For instance, Stations that are geographically close tend to exhibit higher spatial correlations due to similar exposure to prevailing winds and local sources of pollution, i.e stations within urban areas like Munich, tend to show higher spatial correlation which could be influenced by traffic and industrial activities hence have a more uniform impact within smaller, densely populated areas. In contrast, monitoring stations in rural or less densely populated areas, such as ``Andechs/Rothenfeld" and ``Sulzbach-Rosenberg/Lohe," exhibit lower spatial correlation with other stations. This indicates that PM concentrations in these areas are influenced by more localised or unique factors, such as specific agricultural practices or natural terrain features, contributing to spatial variability in pollution levels.

\noindent This relationship suggests that the stations share common local factors or pollution sources in their vicinity, significantly influencing the observed air quality. Spatial dependency may go beyond physical proximity. The high spatial dependence observed is not due to the stations being close geographically but because they share similar dynamics or characteristics of the area they are based. The stations based in urban traffic zones tend to have a similar higher dependency because their pollution factors are similar.  This could include similarities in local sources of pollution, environmental conditions, or other factors influencing air quality. A high level of spatial dependency has important implications, suggesting that interventions or responses triggered by data from one monitoring station may also be relevant and effective for the other station. The results are key in policy formulations in pollution control measures like traffic restrictions and industrial emissions controls within urban clusters.

\section{Conclusion}

In spatiotemporal modelling, accurate specification and/or estimation of the spatial weights matrix is important because it influences the parameter estimates obtained. In this paper, we have proposed an estimation procedure for this spatial weights matrix and other model parameters by adding LASSO penalties to maximum likelihood estimation approach. Through Monte Carlo simulations, the model evaluation metrics have shown reliability and accuracy in performance, specifically in the identification and distinction of zero and non-zero weights across different spatial units. 

\noindent This observed trend of increased estimation accuracy with an increasing number of time points is significant, particularly in the context of modelling PM (particulate matter) concentrations. The LASSO technique has shown accuracy in selecting the spatial weights by shrinking insignificant ones to zero, as evidenced in our results. With the rich temporal dimension of the PM concentrations data, the spatiotemporal autoregressive model accurately captures the spatial and temporal dependencies in these monitoring stations. The study shows that some stations with similar local dynamics tend to have similar patterns, hence showing a high spatial dependence, which helps explain how these PM concentrations differ across stations. Therefore, our paper not only advances the understanding of spatiotemporal dynamics but also offers tangible applications in addressing critical environmental issues.

\noindent Further research on time-varying spatial structure could be done, as some spatial dependencies may rely on the changing times Also, with the use of LASSO, future research could focus on the incorporation of machine learning and deep learning techniques due to their potential to capture hidden patterns in time-varying spatial structures with heightened precision and predictive capabilities. Further research could also focus on reducing the model estimation times, especially for heavily parameterised models.

\newpage
\bibliography{references} 

\begin{thebibliography}{}

\bibitem[Ahrens and Bhattacharjee, 2015]{ahrens2015two}
Ahrens, A. and Bhattacharjee, A. (2015).
\newblock Two-step lasso estimation of the spatial weights matrix.
\newblock {\em Econometrics}, 3(1):128--155.

\bibitem[Ali and Tibshirani, 2019]{ali2019generalized}
Ali, A. and Tibshirani, R.~J. (2019).
\newblock {The Generalized Lasso Problem and Uniqueness}.
\newblock {\em Electronic Journal of Statistics}, 13(2):2307 -- 2347.

\bibitem[Anselin, 1988]{anselin1988spatial}
Anselin, L. (1988).
\newblock {\em Spatial econometrics: methods and models}, volume~4.
\newblock Springer Science \& Business Media.

\bibitem[Bhattacharjee and Jensen-Butler, 2006]{bhattacharjee2006estimation}
Bhattacharjee, A. and Jensen-Butler, C. (2006).
\newblock Estimation of spatial weights matrix, with an application to
  diffusion in housing demand.
\newblock {\em Centre for Research into Industry, Enterprise, Finance, and the
  Firm Discussion Paper}, 519.

\bibitem[Bhattacharjee and Jensen-Butler, 2013]{bhattacharjee2013estimation}
Bhattacharjee, A. and Jensen-Butler, C. (2013).
\newblock Estimation of the spatial weights matrix under structural
  constraints.
\newblock {\em Regional Science and Urban Economics}, 43(4):617--634.

\bibitem[Cressie, 2015]{cressie2015statistics}
Cressie, N. (2015).
\newblock {\em Statistics for spatial data}.
\newblock John Wiley \& Sons.

\bibitem[Debarsy and Le~Gallo, 2025]{debarsy2025identification}
Debarsy, N. and Le~Gallo, J. (2025).
\newblock Identification of spatial spillovers: Do's and don'ts.
\newblock {\em Journal of Economic Surveys}.

\bibitem[Dempster, 1972]{dempster1972covariance}
Dempster, A.~P. (1972).
\newblock Covariance selection.
\newblock {\em Biometrics}, pages 157--175.

\bibitem[Fan and Li, 2001]{fan2001variable}
Fan, J. and Li, R. (2001).
\newblock Variable selection via nonconcave penalized likelihood and its oracle
  properties.
\newblock {\em Journal of the American statistical Association},
  96(456):1348--1360.

\bibitem[Fass{\`o} et~al., 2022]{fasso2022spatiotemporal}
Fass{\`o}, A., Maranzano, P., and Otto, P. (2022).
\newblock Spatiotemporal variable selection and air quality impact assessment
  of covid-19 lockdown.
\newblock {\em Spatial Statistics}, 49:100549.

\bibitem[Friedman et~al., 2010]{friedman2010note}
Friedman, J., Hastie, T., and Tibshirani, R. (2010).
\newblock A note on the group lasso and a sparse group lasso.
\newblock {\em arXiv preprint arXiv:1001.0736}.

\bibitem[Getis and Aldstadt, 2004]{getis2004constructing}
Getis, A. and Aldstadt, J. (2004).
\newblock Constructing the spatial weights matrix using a local statistic.
\newblock {\em Geographical analysis}, 36(2):90--104.

\bibitem[Gibbons and Overman, 2012]{gibbons2012mostly}
Gibbons, S. and Overman, H.~G. (2012).
\newblock Mostly pointless spatial econometrics?
\newblock {\em Journal of regional Science}, 52(2):172--191.

\bibitem[Kelejian and Prucha, 1999]{kelejian1999generalized}
Kelejian, H.~H. and Prucha, I.~R. (1999).
\newblock A generalized moments estimator for the autoregressive parameter in a
  spatial model.
\newblock {\em International economic review}, 40(2):509--533.

\bibitem[Krylova and Okhrin, 2022]{krylova2022managing}
Krylova, M. and Okhrin, Y. (2022).
\newblock Managing air quality: Predicting exceedances of legal limits for pm10
  and o 3 concentration using machine learning methods.
\newblock {\em Environmetrics}, 33(2):e2707.

\bibitem[LeSage and Pace, 2009]{lesage2009introduction}
LeSage, J. and Pace, R.~K. (2009).
\newblock {\em Introduction to spatial econometrics}.
\newblock Chapman and Hall/CRC.

\bibitem[Meinshausen and B{\"u}hlmann, 2006]{meinshausen2006high}
Meinshausen, N. and B{\"u}hlmann, P. (2006).
\newblock {High-dimensional graphs and variable selection with the Lasso}.
\newblock {\em The Annals of Statistics}, 34(3):1436 -- 1462.

\bibitem[Merk and Otto, 2022]{merk2022estimation}
Merk, M.~S. and Otto, P. (2022).
\newblock Estimation of the spatial weighting matrix for regular lattice
  data—an adaptive lasso approach with cross-sectional resampling.
\newblock {\em Environmetrics}, 33(1):e2705.

\bibitem[Otto et~al., 2024a]{otto2024review}
Otto, P., Fass{\`o}, A., and Maranzano, P. (2024a).
\newblock A review of regularised estimation methods and cross-validation in
  spatiotemporal statistics.
\newblock {\em Statistic Surveys}, 18:299--340.

\bibitem[Otto et~al., 2024b]{otto2024spatiotemporal}
Otto, P., Fusta~Moro, A., Rodeschini, J., Shaboviq, Q., Ignaccolo, R., Golini,
  N., Cameletti, M., Maranzano, P., Finazzi, F., and Fass{\`o}, A. (2024b).
\newblock Spatiotemporal modelling of pm 2.5 concentrations in lombardy
  (italy): a comparative study.
\newblock {\em Environmental and Ecological Statistics}, pages 1--28.

\bibitem[Otto and Steinert, 2023]{otto2023estimation}
Otto, P. and Steinert, R. (2023).
\newblock Estimation of the spatial weighting matrix for spatiotemporal data
  under the presence of structural breaks.
\newblock {\em Journal of Computational and Graphical Statistics},
  32(2):696--711.

\bibitem[{The European Environment Agency (EEA)}, 2023]{EEA}
{The European Environment Agency (EEA)} (2023).
\newblock Air pollution.
\newblock Accessed: August 9, 2023.

\bibitem[Tibshirani, 1996]{tibshirani1996regression}
Tibshirani, R. (1996).
\newblock Regression shrinkage and selection via the lasso.
\newblock {\em Journal of the Royal Statistical Society: Series B
  (Methodological)}, 58(1):267--288.

\bibitem[Tibshirani, 2013]{tibshirani2013lasso}
Tibshirani, R.~J. (2013).
\newblock {The lasso problem and uniqueness}.
\newblock {\em Electronic Journal of Statistics}, 7(none):1456 -- 1490.

\bibitem[{US Environmental Protection Agency}, 2023]{EPA_PM}
{US Environmental Protection Agency} (2023).
\newblock Health and environmental effects of particulate matter (pm).
\newblock Accessed: August 9, 2023.

\bibitem[Wang and Leng, 2008]{wang2008note}
Wang, H. and Leng, C. (2008).
\newblock A note on adaptive group lasso.
\newblock {\em Computational statistics \& data analysis}, 52(12):5277--5286.

\bibitem[Ward and Gleditsch, 2018]{ward2018spatial}
Ward, M.~D. and Gleditsch, K.~S. (2018).
\newblock {\em Spatial regression models}, volume 155.
\newblock Sage Publications.

\bibitem[Wright, 2006]{wright2006numerical}
Wright, S.~J. (2006).
\newblock Numerical optimization.

\bibitem[Yu et~al., 2008]{yu2008quasi}
Yu, J., De~Jong, R., and Lee, L.-f. (2008).
\newblock Quasi-maximum likelihood estimators for spatial dynamic panel data
  with fixed effects when both n and t are large.
\newblock {\em Journal of Econometrics}, 146(1):118--134.

\bibitem[Zhang, 2010]{zhang2010nearly}
Zhang, C.-H. (2010).
\newblock {Nearly unbiased variable selection under minimax concave penalty}.
\newblock {\em The Annals of Statistics}, 38(2):894 -- 942.

\bibitem[Zhang and Huang, 2008]{zhang2008sparsity}
Zhang, C.-H. and Huang, J. (2008).
\newblock {The sparsity and bias of the Lasso selection in high-dimensional
  linear regression}.
\newblock {\em The Annals of Statistics}, 36(4):1567 -- 1594.

\bibitem[Zhang, 2009]{zhang2009some}
Zhang, T. (2009).
\newblock {Some sharp performance bounds for least squares regression with L1
  regularization}.
\newblock {\em The Annals of Statistics}, 37(5A):2109 -- 2144.

\bibitem[Zhang and Yu, 2018]{zhang2018spatial}
Zhang, X. and Yu, J. (2018).
\newblock Spatial weights matrix selection and model averaging for spatial
  autoregressive models.
\newblock {\em Journal of Econometrics}, 203(1):1--18.

\bibitem[Zhao and Yu, 2006]{zhao2006model}
Zhao, P. and Yu, B. (2006).
\newblock On model selection consistency of lasso.
\newblock {\em The Journal of Machine Learning Research}, 7:2541--2563.

\end{thebibliography}

\newpage
\section{Appendix}

\begin{figure}[h!]
  \centering
  \includegraphics[width=\textwidth]{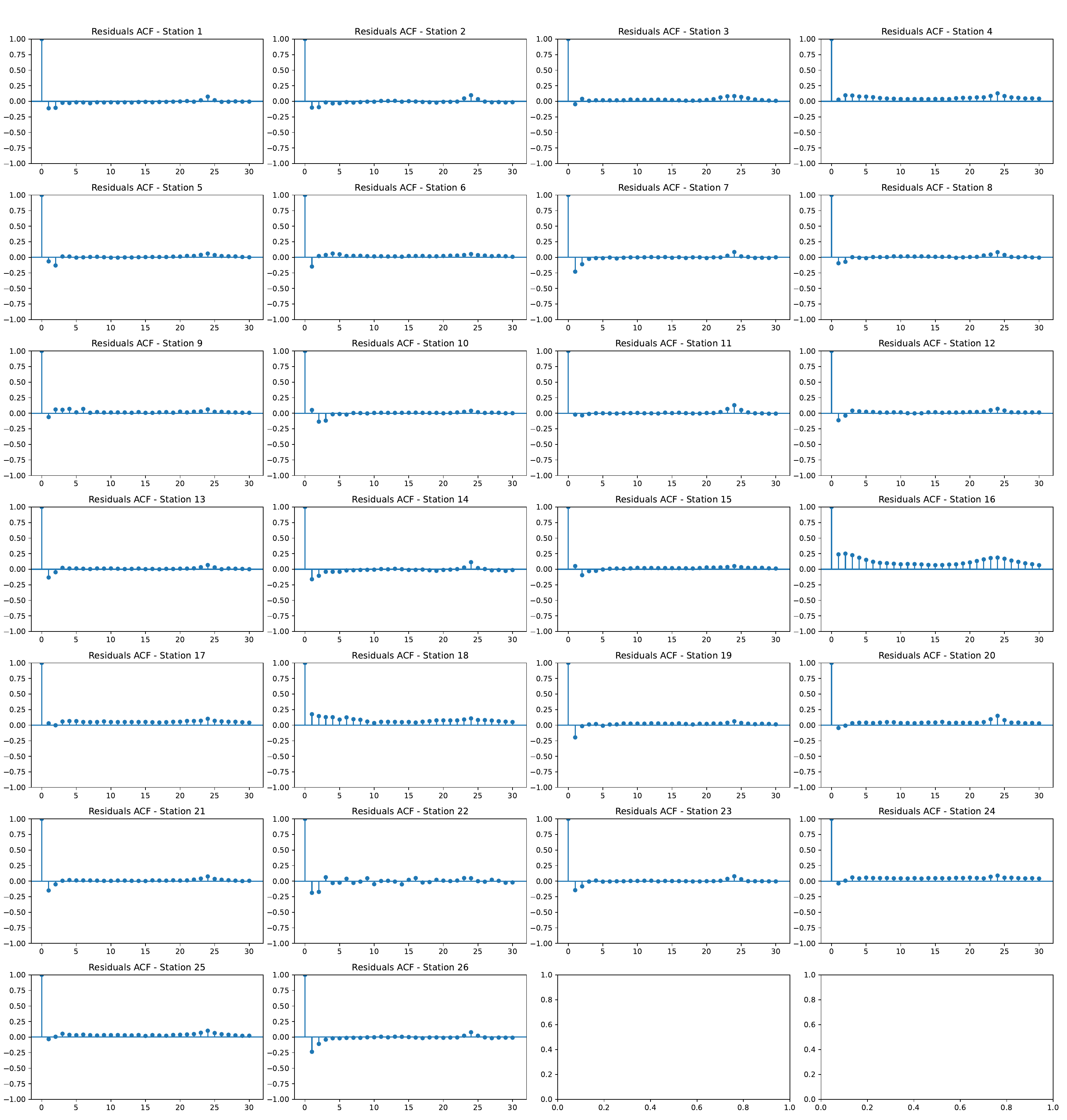}
  \caption{Autocorrelation Function (ACF) plots}\label{fig:Autocorrelation plot}
\end{figure}
\end{document}